\documentclass[reqno,a4paper,10pt]{amsproc}

\usepackage{graphicx}         
\usepackage{amsfonts}
\usepackage{amssymb}
\usepackage{eucal}
\usepackage[latin1]{inputenc}
\usepackage[all]{xy}

\newtheorem{theorem}{Theorem}[section]
\newtheorem{corollary}[theorem]{Corollary}
\newtheorem{lemma}[theorem]{Lemma}
\newtheorem{proposition}[theorem]{Proposition}
\theoremstyle{definition}

\newtheorem{remark}[theorem]{Remark}

\newcommand{\eps}{\epsilon}

\renewcommand{\l}{\lambda}

\newcommand{\PB}{\left\{\cdot\,,\cdot\right\}}

\newcommand{\pb}[1]{\left\{#1\right\}}

\renewcommand{\(}{\left(}
\renewcommand{\)}{\right)}

\newcommand{\set}[1]{\left\{#1\right\}}

\newcommand{\bbR}{\mathbb R}

\newcommand{\ds}{\displaystyle}

\newcommand{\leqs}{\leqslant}
\newcommand{\geqs}{\geqslant}
\newcommand{\pp}[2]{\frac{\partial#1}{\partial#2}}
\newcommand{\p}{\partial}

\newcommand{\Rk}{\hbox{rank\,}}
\newcommand{\diff}{{\rm d }}

\renewcommand{\leq}{\leqs}

\newif\ifprivate
\privatefalse

 \numberwithin{equation}{section}

\def\???{\ifprivate {\bf {???}} \marginpar{{\Huge {\bf ?}}}\else \fi}
\numberwithin{equation}{section}

\begin{document}  

\nocite{*} 

\parskip 4pt
\baselineskip 16pt


\title[Generalized Lotka-Volterra systems]
{Liouville integrability and superintegrability of  a generalized Lotka-Volterra system and its Kahan
discretization}

\author[Kouloukas, Quispel, and Vanhaecke]{Theodoros E. Kouloukas$^{1,3}$, G. R. W. Quispel$^{1}$ and Pol Vanhaecke$^{2}$ 
\\
\\$^{1}$Department of Mathematics and Statistics,
\\
La Trobe University, Bundoora VIC 3086, Australia \\
$^{2}$Laboratoire de Math\'ematiques et Applications, \\  
UMR 7348 du CNRS, Universit\'e de Poitiers,  \\ 
 86962 Futuroscope Chasseneuil
Cedex, France \\ 
$^{3}$SMSAS, University of Kent, Canterbury, UK}

\email{theodoroskouloukas@gmail.com}

\email{R.Quispel@latrobe.edu.au}

\email{pol.vanhaecke@math.univ-poitiers.fr}



\begin{abstract}

We prove the Liouville and superintegrability of a generalized Lotka-Volterra system and its Kahan
discretization.

\end{abstract}

\maketitle

\tableofcontents

\section{Introduction}

The Kac-van Moerbeke system is a prime example of an integrable system, described by the differential equations
\begin{equation}\label{eq:KM_intro}
  \dot x_i=x_i(x_{i+1}-x_{i-1})\,,\qquad (i=1,\dots,n)\;,
\end{equation}%
where $x_0=x_{n+1}=0$. It was first introduced and studied, together with some of its generalizations, by Lotka to
model oscillating chemical reactions and by Volterra to describe population evolution in a hierarchical system of
competing species (see \cite{Lotka, Volterra}). By now, many generalizations of (\ref{eq:KM_intro}) have been
introduced and studied, often from the point of (Liouville or algebraic) integrability \cite{Bog2,Fer,FV} or Lie
theory \cite{Bog2,dam}, but also in relation with other integrable systems \cite{DF,Moser,KKQTV}. In our recent
study \cite{KKQTV}, a natural generalization of (\ref{eq:KM_intro}) came up in the study of a class of multi-sums
of products: we considered the system
\begin{equation}\label{eq:fat_LV_ode}
  \dot x_i=x_i\left(\sum_{j>i}x_{j}-\sum_{j<i}x_{j}\right)\,,\qquad (i=1,\dots,n)\;,
\end{equation}%
we showed its Liouville and superintegrability and we used it to show the Liouville and superintegrability (or
non-commutative integrability) of the Hamiltonian system defined by the above-mentionned class of functions. The
system (\ref{eq:fat_LV_ode}) has a Hamiltonian structure, described by the Hamiltonian function and Poisson
structure, which are respectively given by
\begin{equation}\label{eq:fat_LV_ham}
  H=\sum_{i=1}^n x_i\;,\qquad \pb{x_i,x_j}=x_ix_j\;,\quad (i<j)\;.
\end{equation}%
We consider in the present paper the case of a general linear Hamiltonian 
\begin{equation}\label{eq:GLV_ham}
  H=\sum_{i=1}^n a_ix_i\;,
\end{equation}%
with the Poisson structure still given by (\ref{eq:fat_LV_ham}). The differential equations which describe this
Hamiltonian system are given by
\begin{equation}\label{eq:GLV_ode}
  \dot x_i=x_i\left(\sum_{j>i}a_jx_{j}-\sum_{j<i}a_jx_{j}\right)\,,\qquad (i=1,\dots,n)\;.
\end{equation}%
When all the parameters $a_i$ are different from zero, a trivial rescaling (which preserves the Poisson structure)
leads us back to (\ref{eq:fat_LV_ham}), so the novelty of our study is mainly concerned with the case where at
least one (but not all!) of the parameters $a_i$ is zero, though all results below are also valid in case all the
parameters $a_i$ are different from zero. By explicitly exhibiting a set of $[(n+1)/2]$ involutive (Poisson
commuting) rational functions, which are shown to be functionally independent, we show that (\ref{eq:GLV_ham}) is
Liouville integrable (Theorems \ref{thm:integrable} and \ref{thm:integrable_odd}). We also exhibit $n-1$
functionally independent first integrals, thereby showing that (\ref{eq:GLV_ode}) is superintegrable (Theorem
\ref{thm:ode_super}).  Finally, we construct for any initial conditions explicit solutions of (\ref{eq:GLV_ode})
(Proposition \ref{prp:ode_solutions}).

In Section \ref{sec:discrete} we study the Kahan discretization (see e.g.\ \cite{MR3002854}) of
(\ref{eq:GLV_ode}), which we explicitly describe (Proposition \ref{prop:kahan}). We also show that the map
defined by the Kahan discretization is a Poisson map (Proposition \ref{prop:kahan_poisson}). Upon comparing the
latter map with the solutions to the continuous system (\ref{eq:GLV_ode}), we prove that the Kahan map is a time
advance map for this Hamiltonian system, and we derive from it that the discrete system is both Liouville and
superintegrable, with the same first integrals as the continuous system (Proposition \ref{prop:discrete_sol} and
Corollary \ref{cor:kahan_integrable}). 

We finish the paper with some comments and perspectives for future work (Section \ref{sec:comments}).

\section{A generalized Lotka-Volterra system}
Let $n$ be an arbitrary positive integer. We consider on $\bbR^n$ the generalized Lotka-Volterra system
\begin{equation} \label{system1}
  \dot{x_i}=x_i \sum_{j=1}^n A_{ij} x_j\;, \qquad (i=1,\dots,n)\;,
\end{equation}
where $A$ is the square matrix
\begin{equation}\label{Amatrix}
  A=\begin{pmatrix}
  0 & a_2& a_3& \dots&a_n \\
  -a_1 & 0& a_3& \dots&a_n \\
  -a_1 & -a_2& 0& \dots&a_n \\
  \vdots & \vdots & \vdots & \ddots \\
  -a_1 &-a_2&-a_3& \dots&0
\end{pmatrix}\;, 
\end{equation}
and $(a_1,\dots,a_n) \in \mathbb{R}^n\setminus\set{(0,\dots,0)}$. Like most Lotka-Volterra system, it has a linear
function as Hamiltonian, to wit $H:=a_1 x_1+a_2 x_2+\cdots+a_n x_n$; the corresponding (quadratic) Poisson
structure is defined by the brackets $\{x_i,x_j\}:= x_i x_j$, for $1\leq i < j \leq n$. The following elementary
lemma, which will play a key rôle in the proof of Theorem \ref{thm:integrable} below, shows that rescaling the
parameters $a_i$ by non-zero constants leads to isomorphic Hamiltonian systems.
\begin{lemma}\label{lma:rescaling}
  Let $c_1,\dots,c_n$ be arbitrary non-zero real constants. Then the linear change of coordinates $x_i\mapsto
  x_i/c_i$ transforms the generalized Lotka-Volterra system with parameters $a_1,\dots,a_n$ into the generalized
  Lotka-Volterra system with parameters $a_1c_1,\dots,a_nc_n$.
\end{lemma}
\begin{proof}
Let $y_i:=x_i/c_i$. Then $\pb{y_i,y_j}=\pb{x_i,x_j}/(c_ic_j)=\frac{x_ix_j}{c_ic_j}=y_iy_j$, for any $i<j$, which
shows that the change of variables preserves the Poisson structure. Clearly, in terms of the new variables, the
Hamiltonian reads $H=a_1c_1y_1+\cdots+a_nc_ny_n$, which is the Hamiltonian of the generalized Lotka-Volterra system
with constants $a_ic_i$.
\end{proof}
As an application of the lemma, we have that when the parameters $a_i$ are all non-zero, we can rescale them all to
1, and (\ref{system1}) becomes (\ref{eq:fat_LV_ode}) (which is system (3.5) in \cite{KKQTV}). In this case, the
matrix $A$ is skew-symmetric and so (\ref{system1}) is a genuine Lotka-Volterra system, whose Liouville and
superintegrability have extensively been studied in \cite{KKQTV}. When some of the parameters $a_i$ are zero, we
get new (non-isomorphic) systems. As we will show in this section, all these systems are Liouville and
superintegrable.

For the study of the general case, it is convenient to introduce the functions $v_i:=a_1 x_1+\cdots+a_i x_i$, for
$i=1,\dots,n$; we also set $v_0:=0$. In terms of these functions, $H=v_n$ and the system (\ref{system1}) can
equivalently be written as
\begin{equation} \label{system2}
  \dot{x_i}=x_i(H-v_i-v_{i-1})\;,\qquad  (i=1,\dots,n)\;.
\end{equation}
For $i<j$, one has $\pb{v_i,x_j}=v_ix_j,$ and so the Poisson brackets of the functions $v_i$ are given by
\begin{equation} \label{pb_vi}
  \pb{v_i,v_j}=v_i(v_j-v_i)\;,\qquad (i<j)\;.
\end{equation}
In particular, remembering that $H=v_n$,
\begin{equation}\label{eq:ode_for_v}
  \dot v_i=\pb{v_i,H}=v_i(H-v_i)\;,
\end{equation}
for $i=1,\dots,n$. If $a_1 a_2\dots a_n \neq 0$, the functions $v_i$ define new coordinates on~$\mathbb{R}^n$,
since then $x_{k}=(v_{k}-v_{k-1})/a_{k}$ for $k=1,\dots,n$; moreover, the system (\ref{system1}) totally decouples
in terms of these coordinates since it takes the simple form $\dot{v_i}=v_i(H-v_i)$, for $i=1,\dots,n$. However,
the functions $v_i$ do not define coordinates when at least one of the $a_k$ is zero, because if $a_k=0$ then
$v_k=v_{k-1}$.

With a view to proving Liouville integrability, we define for $k=1,\dots,\left[\frac{n}{2}\right]$ the functions
\begin{equation} \label{JF}
  J_k:=\frac{x_1 x_3 \dots x_{2k-1}}{x_2 x_4 \dots x_{2k}}\;,
\end{equation}
and for $k=1,\dots,\left[\frac{n+1}{2}\right]$ the functions
\begin{equation}\label{E:integrals}
  F_k:=
    \left\{ \begin{array}{ll}
      v_{2k-1}\ds\frac{x_{2k+1}x_{2k+3}\ldots x_{n}}{x_{2k}x_{2k+2}\ldots x_{n-1}}&\mbox{ if}\  n\mbox{ is odd},\\
      \\
      v_{2k}\ds\frac{x_{2k+2}x_{2k+4}\ldots x_n}{x_{2k+1} x_{2k+3}\ldots x_{n-1}}&\mbox{ if} \  n\mbox{ is even}.
    \end{array} \right.
  \end{equation}
Notice that $F_{[(n+1)/2]}=v_n=H$, the Hamiltonian (\ref{eq:GLV_ham}). For odd $n$, we also introduce the function
\begin{equation*}
  C:=\frac{x_1 x_3 \dots x_{n}}{x_2 x_4 \dots x_{n-1}}\;.
\end{equation*}
\begin{proposition} \label{invol}
  For any $k,l\in \{1,\dots,\left[\frac{n}{2}\right]\}$, 
  \begin{equation} \label{inv}
    \{J_k,J_l\}=\{F_k,F_l\}=\pb{F_k,H}=0\;.
  \end{equation}
  Moreover, when $n$ is odd, $C$ is a Casimir function of the Poisson bracket $\PB$.
\end{proposition}
\begin{proof}
First, we notice that for any $k=1,\dots,\left[\frac{n}{2}\right]$
\begin{equation} \label{prtder}
  x_i \frac{\partial J_k}{\partial x_i}= \left\{
    \begin{array}{cc}
         (-1)^{i+1}J_k  & \mbox{for } 1 \leq i \leq 2k\;,  \\
          0  & \mbox{for } 2k < i \leq n\;.
    \end{array}
    \right.
\end{equation}
It follows that, for $k<l \in \{1,\dots,\left[\frac{n}{2}\right]\}$, we have
\begin{eqnarray*}
  \{J_k,J_l \}&=&\sum_{1 \leq i < j \leq n} x_i x_j
  \(\frac{\partial{J_k}}{\partial{x_i}}\frac{\partial{J_l}}{\partial{x_j}}-
  \frac{\partial{J_k}}{\partial{x_j}}\frac{\partial{J_l}}{\partial{x_i}}\)\\
  &=&\sum_{1 \leq i < j \leq 2k}[ (-1)^{i+1}J_k (-1)^{j+1}J_l-(-1)^{j+1}J_k (-1)^{i+1}J_l ] \\
  &&+ \sum_{1 \leq i \leq 2k < j \leq 2l}(-1)^{i+1}J_k (-1)^{j+1}J_l=0\;.
\end{eqnarray*}
%
This shows the first equality of (\ref{inv}). We show the two other equalities of (\ref{inv}) for even $n$. To do
this, it suffices to show that $\pb{F_k,F_l}=0$ for $1\leq k<l\leq n/2$ since $F_{n/2}=H$.  We set $F_k=v_{2k}I_k$,
i.e., we define $I_k$ by
\begin{equation*}
  I_k:=\ds\frac{x_{2k+2}x_{2k+4}\ldots x_n}{x_{2k+1} x_{2k+3}\ldots x_{n-1}}\;.
\end{equation*}%
As in (\ref{prtder}), we have that
\begin{equation} \label{prtder2}
  x_i \frac{\partial I_k}{\partial x_i}= \left\{
  \begin{array}{cc}
    0  & \mbox{for } 1 \leq i \leq 2k\;,\\
    (-1)^iI_k  & \mbox{for } 2k < i \leq n\;, 
  \end{array}
  \right.
\end{equation}
from which it follows, as above, that $\pb{I_k,I_l}=0$ and that $\pb{I_k,J_l}=0$ for all $k,l\in \{1,\dots,n/2\}$.
Also, for any $j \in \{ 1,\dots,n \}$
\begin{eqnarray*}
  \pb{I_k,x_j}=\sum_{i=1}^n\pp{I_k}{x_i}\pb{x_i,x_j}
  =\left(\sum_{1\leq i < j} x_i \frac{\partial{I_k}}{\partial{x_i}}-
  \sum_{j < i\leq n} x_i\frac{\partial{I_k}}{\partial{x_i}}\right)x_j
\end{eqnarray*}
and using (\ref{prtder2}) we derive that
\begin{equation} \label{brIx}
  \{I_k,x_j\}= \left\{
    \begin{array}{cc}
          0 &\mbox{for } j \leq 2k\;, \\
         -I_kx_j &\mbox{for } 2k<j\;,
    \end{array}
  \right.
  \mbox { and }
  \{I_k,v_j\}= \left\{
    \begin{array}{cc}
          0&\mbox{for } j\leq 2k\;, \\
         -I_k(v_j-v_{2k})  &\mbox{for } 2k<j\;.
    \end{array}
  \right.
\end{equation}
It follows from (\ref{pb_vi}) and (\ref{brIx}) that, for any $k<l\leq n/2$,
\begin{eqnarray*}
  \{F_k,F_l\}&=&\{v_{2k}I_k,v_{2l}I_l\}=v_{2k}I_l\{I_k,v_{2l}\}+v_{2l}I_k \{v_{2k},I_l \}+I_kI_l\{v_{2k},v_{2l}\}\\
  &=&-v_{2k}I_kI_l(v_{2l}-v_{2k})+0+v_{2k}I_kI_l(v_{2l}-v_{2k})=0\;.
\end{eqnarray*}
This shows the second half of (\ref{inv}) for $n$ even; for $n$ odd, the proof is very similar (in this case,
$H=F_{(n+1)/2}$ and one proves as above that $\pb{F_k,F_l}=0$ for $1\leq k<l\leq (n+1)/2$). Finally we show that
$C$ is a Casimir function (when $n$ is odd). For $j=1,\dots,n$,
\begin{eqnarray*}
  \{C,x_j\}&=&\sum_{i=1}^n\pp C{x_i}\pb{x_i,x_j}=
     \left(\sum_{1\leq i < j} x_i \frac{\partial{C}}{\partial{x_i}}-
    \sum_{j < i\leq n} x_i\frac{\partial{C}}{\partial{x_i}}\right)x_j\\
  &=&\sum_{1\leq i < j}(-1)^{i+1}Cx_j-\sum_{j < i\leq n} (-1)^{i+1}Cx_j=0\;,
\end{eqnarray*}
which shows our claim.
\end{proof}

\begin{theorem}\label{thm:integrable}
  Suppose that $n$ is even. Let $\ell$ denote the smallest integer such that $a_{\ell+1}\neq 0$ (in particular,
  $\ell=0$ when $a_1\neq0$) and let $\lambda:=\left[\frac{\ell}{2} \right]$. The $\frac{n}{2}$ functions
  $J_1,J_2,\dots,J_{\lambda},H,F_{\lambda+1},F_{\lambda+2},\dots,F_{\frac{n}{2}-1}$ are pairwise in involution and
  functionally independent, hence they define a Liouville integrable system on $(\bbR^n,\PB)$.
\end{theorem}

\begin{proof}
We know already from Proposition \ref{invol} that the functions $J_k$ are pairwise in involution, and also the
functions $F_l$ (recall that $F_{n/2}=H$). We show that $\{J_k,F_l\}=0$ for $k=1,\dots, \lambda$ and
$l=\lambda+1,\dots,\frac{n}{2}$. To do this, we use the following analog of (\ref{brIx}), which is easily obtained
from (\ref{prtder}):
\begin{equation*}
  \{J_k,v_j\}= \left\{
    \begin{array}{cc}
          J_k v_j &\mbox{for } j \leq 2k\;, \\
          J_kv_{2k} &\mbox{for } 2k<j\;.
    \end{array}
  \right.
\end{equation*}
It follows that, for the above values of $k,l$, which satisfy $k\leq \l<l$, one has $\pb{J_k,v_{2l}}=J_kv_{2k}=0$
(the last equality follows from $2k\leq 2\l\leq\ell$ and $v_i=0$ for $i\leq\ell$), and so
\begin{eqnarray*}
  \{J_k,F_l\}= \{J_k,v_{2l}I_l\}=v_{2k}\{J_k,I_l\}+I_l\{J_k,v_{2l}\}=0\;;
\end{eqnarray*}
in the last step we also used that the functions $I_i$ and $J_j$ are in involution (see the proof of Proposition
\ref{invol}). This shows that the $\frac{n}{2}$ functions
\begin{equation}\label{eq:fun_even}
  J_1,J_2,\dots,J_{\lambda},H,F_{\lambda+1},F_{\lambda+2},\dots,F_{\frac{n}{2}-1}
\end{equation}
are pairwise in involution.

We now show that these functions are functionally independent. We first do this when all $a_i$ are zero, except for
$a_{\ell+1}$ which we may suppose to be equal to~1; then $v_i=x_{\ell+1}=H$ for $i>\ell$ and $v_i=0$ for $i\leq
\ell$. The Jacobian matrix of the above functions (\ref{eq:fun_even}) with respect to $x_1,\dots,x_n$ (in that
order) is easily seen to have the following block form:
$$
  Jac=
  \begin{pmatrix}
    A&0&0\\
    0&1&0\\
    0&\star&B
  \end{pmatrix}\;,
$$
where $A$ has size $\l\times\ell$ and $B$ has size $(\frac n2-\l-1)\times (n-\ell-1)$.  We show that this matrix
has full rank $n/2$ (which is equal to the number of rows of $Jac$). To do this, it is sufficient to show that $A$
has full rank $\lambda$ and that $B$ has full rank $n/2-\lambda-1$ (the value of the column vector $\star$ is
irrelevant). Consider the square submatrix $A'$ of $A$ consisting only of its even-numbered columns. For $k<l$ we
have $A'_{kl}=A_{k,2l}=\p J_k/\p x_{2l}=0$, since $J_k$ only depends on $x_1,\dots,x_{2k}$. It follows that $A'$ is
a lower triangular matrix. Moreover, $A'_{kk}=A_{k,2k}=\p J_k/\p x_{2k}\neq0$, hence $A'$ is non-singular. This
shows that $\Rk(A)=\Rk(A')=\lambda$. Similarly, we extract from $B$ a square submatrix $B'$ by selecting from $B$
its even-numbered (respectively odd-numbered) columns when $\ell$ is even (respectively odd). For $k>l$ we have
$B'_{kl}= \p F_{\l+k}/\p x_{2\l+1+2l}=\p(v_{2\l+2k}I_{\l+k})/\p x_{2\l+1+2l}=\p(x_{\ell+1}I_{\l+k})/\p
x_{2\l+1+2l}=x_{\ell+1}\p I_{\l+k}/\p x_{2\l+1+2l}=0$, since $I_{\l+k}$ is independent of
$x_1,\dots,x_{2\l+2k}$. However, $B'_{kk}=x_{\ell+1}\p I_{\l+k}/\p x_{2\l+1+2k}\neq0$, because $I_{\l+k}$ does
depend on $x_{2\l+1+2k}$. This shows that $B'$ is a non-singular upper triangular matrix, hence
$\Rk(B)=\Rk(B')=n/2-\lambda-1$. We have thereby shown that if $H=x_{\ell+1}$, then the $n/2$ functions in
(\ref{eq:fun_even}) are functionally independent; since the rank of the Poisson structure $\PB$ is $n$, we have
shown Liouville integrability in this case.

We now consider the general case, where several of the $a_i$ may be non-zero. We may still suppose that
$a_{\ell+1}=1$; as above, $a_1=\dots=a_\ell=0$. Let us view $a_{\ell+2},\dots,a_n$ as arbitrary parameters and
consider the matrix
$$
  Jac':=
  \begin{pmatrix}
    A'&0&0\\
    0&1&0\\
    0&\star&B'
  \end{pmatrix}\;, 
$$
where $A'$ and $B'$ are square matrices which are constructed as in the previous paragraph. It depends polynomially
on the parameters $a_{\ell+2},\dots,a_n$ and we have shown that the determinant of $Jac'$ is non-zero when we set
all the parameters $a_{\ell+2},\dots,a_n$ equal to zero. By continuity, the determinant remains non-zero when the
parameters $a_{\ell+2},\dots,a_n$ are sufficiently close to zero, which proves that the $n/2$ functions in
(\ref{eq:fun_even}) are functionally independent for such values of the parameters. In view of Lemma
\ref{lma:rescaling}, any non-zero rescaling of the parameters leads to isomorphic systems, so for any values of
$a_{\ell+2},\dots,a_n$, the functions in (\ref{eq:fun_even}) are functionally independent. This shows Liouville
integrability for any values of the parameters $a_1,\dots,a_n$.
\end{proof}

When $n$ is odd, the rank of the Poisson structure $\PB$ is $n-1$, so for Liouville integrability we need $(n+1)/2$
functionally independent functions in involution. Recall from Proposition \ref{invol} that in this case $C$ is a
Casimir function. The Liouville integrability is in this case given by the following theorem, whose proof is
omitted because it is very similar to the proof of Theorem \ref{thm:integrable}.

\begin{theorem}\label{thm:integrable_odd}
  Suppose that $n$ is odd. As before, let $\ell$ denote the smallest integer such that $a_{\ell+1}\neq 0$ and let
  $\lambda:=\left[\frac{\ell}{2} \right]$. The $\frac{n+1}{2}$ functions
  $J_1,J_2,\dots,J_{\lambda},H,F_{\lambda+2},$ $F_{\lambda+3},\dots,F_{\frac{n-1}{2}},C$ are pairwise in involution
  and functionally independent, hence define a Liouville integrable system on $(\bbR^n,\PB)$.
\end{theorem}

We show in the following theorem that the Hamiltonian vector field defined by $H$ is also superintegrable.

\begin{theorem}\label{thm:ode_super}
  The Hamiltonian system (\ref{eq:GLV_ode}) has $n-1$ functionally independent first integrals, hence is
  superintegrable.
\end{theorem}
\begin{proof}
We denote, as before, by $\ell$ the smallest integer such that $a_{\ell+1}\neq 0$ (in particular, $\ell=0$ when
$a_1\neq0$). Suppose first that $a_{\ell+1}$ is the only $a_i$ which is different from zero; by a simple rescaling,
we may assume $a_{\ell+1}=1$, so that $H=x_{\ell+1}$. Then the equations of motion (\ref{eq:GLV_ode}) take the
following simple form:
\begin{equation}\label{eq:super_simple}
  \dot x_i=\left\{
  \begin{array}{cl}
    x_iH\qquad& i\leqs\ell\;,\\
    0\qquad& i=\ell+1\;,\\
    -x_iH\qquad& i>\ell+1\;.\\
  \end{array}
  \right.
\end{equation}%
When $\ell=0$, a complete set of $n-1$ independent first integrals of (\ref{eq:super_simple}) is given by $H=x_1$
and $x_i/x_2,\ (i=3,\dots,n)$. When $\ell\neq0$, we can take besides the Hamiltonian $H=x_{\ell+1}$ the functions
$x_i/x_1,\ (i=2,\dots,\ell)$ and $x_1x_i,\ (i=\ell+2,\dots,n)$.

In the general case, we partition the set $\{1,2,\dots,n\}$ into three subsets ($A$ or $C$ may be empty):
\begin{eqnarray*}
  A&:=&\{1,2,\dots,\ell\}\;,\\
  B&:=&\{i\mid a_i\neq0\}\;,\\
  C&:=&\{i\mid i>\ell+1\hbox{ and } a_i=0\}\;.
\end{eqnarray*}
Since we have treated the case $\#B=1$, we may henceforth assume that $\#B\geqs 2$. Notice that each function $v_i$
(and in particular $H$) depends only on the variables $x_i$ with $i\in B$. It follows that the differential
equations (\ref{system2}),
\begin{equation*}%
  \dot{x_i}=x_i(H-v_i-v_{i-1})\;,\qquad  (i\in B)\;,
\end{equation*}%
involve only the variables $x_j$ with $j\in B$, so they form a subsystem which is the same as the original system,
but now of dimension $m:=\# B$, and with all parameters $a_i,\ i\in B$ different from zero. As explained above (see
Lemma \ref{lma:rescaling} and the remarks which follow its proof) this subsystem is by a simple rescaling
isomorphic to the system (\ref{eq:fat_LV_ode}), for which we know from \cite{KKQTV} that it is superintegrable,
with $m-1$ first integrals which we denote here by $G_1,\dots,G_{m-1}$. We do not need here the precise formulas
for these functions, but only the fact that they depend only on the variables $x_j$ with $j\in B$; this obvious
fact implies that the functions $G_1,\dots,G_{m-1}$ are first integrals of the full system (\ref{eq:GLV_ode}) as
well. Consider, for $i\in A\cup C$ the following rational function:
\begin{eqnarray*}
  K_i:=\left\{
    \begin{array}{ll}
      \ds\frac{(H-a_{\ell+1}x_{\ell+1})x_i}{x_{\ell+1}}\;,&i\in A\;,\\
      \ds\frac{(H-a_{\ell+1}x_{\ell+1})v_{i}^2}{x_ix_{\ell+1}}\;,&i\in C\;.
    \end{array}
  \right.
\end{eqnarray*}
Notice that $H-a_{\ell+1}x_{\ell+1}$ is different from zero, because $\#B\geqs 2$. For $i\in A$, we have that
\begin{eqnarray*}
  (\ln K_i)^\cdot&=&(\ln(H-a_{\ell+1}x_{\ell+1}))^\cdot+(\ln(x_i/x_{\ell+1}))^\cdot\\
    &=&-\frac{a_{\ell+1}\dot x_{\ell+1}}{H-a_{\ell+1}x_{\ell+1}}+a_{\ell+1}x_{\ell+1}=0\;.
\end{eqnarray*}
Indeed, $\dot x_{\ell+1}=x_{\ell+1}(H-v_{\ell+1}-v_{\ell})=x_{\ell+1} (H-a_{\ell+1}x_{\ell+1})$.  Similarly, for
$i\in C$, we have from (\ref{system2}) and (\ref{eq:ode_for_v}) that
\begin{eqnarray*}
  (\ln K_i)^\cdot&=&(\ln(H-a_{\ell+1}x_{\ell+1}))^\cdot+2(\ln v_{i})^\cdot-(\ln(x_ix_{\ell+1}))^\cdot\\
    &=&-a_{\ell+1}x_{\ell+1}+2(H-v_{i})-(H-2v_{i})-(H-a_{\ell+1}x_{\ell+1})=0\;.
\end{eqnarray*}
This shows that the $n-1$ functions $G_1,\dots,G_{m-1}$ and $K_i$, $i\in A\cup C,$ are first integrals of
(\ref{eq:GLV_ode}). Recall that the functionally independent functions $G_1,\dots,G_{m-1}$ depend on $x_i$ with
$i\in B$ only and notice that for $i\in A\cup C$ the variable $x_i$ appears only in $K_i$. It follows that these
$n-1$ first integrals of (\ref{eq:GLV_ode}) are functionally independent, hence (\ref{eq:GLV_ode}) is
superintegrable.
\end{proof}

Finally, we compute the solution $x(t)$ of (\ref{system1}) which corresponds to any given initial condition
$x^{(0)}=(x_1^{(0)},\dots,x_n^{(0)})$. We also introduce the derived functions $v_i(t)=a_1x_1(t)+\cdots+
a_ix_i(t)$, for $i=1,\dots,n$. We denote by $h_0$ the value of the Hamiltonian $H$ at the initial condition
$x^{(0)}$ and we denote $v_i^{(0)}:=v_i(0)$. It follows from (\ref{system2}) and (\ref{eq:ode_for_v}) that we need
to solve
\begin{equation} \label{system2_bis}
  \frac{\diff x_i}{\diff t}(t)=x_i(t)(h_0-v_i(t)-v_{i-1}(t))\;,\qquad  (i=1,\dots,n)\;,
\end{equation}
where
\begin{equation}\label{eq:ode_for_v_bis}
  \frac{\diff v_i}{\diff t}(t)=v_i(t)(h_0-v_i(t))\;,\qquad  (i=1,\dots,n)\;.
\end{equation}
When $v_i^{(0)}=0$, the latter equation has $v_i(t)=0$ as its unique solution; otherwise (\ref{eq:ode_for_v_bis})
is easily integrated by a separation of variables, giving
\begin{equation}\label{eq:sol_a}
  v_i(t)=\frac1{\frac1{h_0}+C_ie^{-h_0t}}\;,\quad\hbox{ or } \quad v_i(t)=\frac1{t+C_i'}\;,
\end{equation}%
depending on whether $h_0\neq 0$ or $h_0=0$. The integrating constants $C_i$ and $C'_i$ are computed from
$v_i(0)=v_i^{(0)}$, which leads to
\begin{equation*}%
  C_i=\frac1{v_i^{(0)}}-\frac1{h_0}\;,\quad\hbox{ and }\quad C_i'=\frac1{v_i^{(0)}}\;.
\end{equation*}%
The functions $v_i(t)$ in (\ref{eq:sol_a}) have very simple primitives, to wit
\begin{equation}\label{eq:primitives}
  \int v_i(t)\diff t= \ln\(\frac{e^{h_0t}}{h_0}+C_i\)\;,\qquad\hbox{ or } \quad \int v_i(t)\diff t=\ln(t+C_i')\;.
\end{equation}%
Substituted in (\ref{system2_bis}), which we write now as $\frac{\diff\ln x_i}{\diff t}(t)=h_0-v_i(t)-v_{i-1}(t)$,
we obtain by integration and by using the primitives (\ref{eq:primitives}) (or $\int v_i(t)\diff t=$ constant in
case $v_i^{(0)}=0$) and the initial condition $x_i(0)=x_i^{(0)}$, the following result:
\begin{proposition}\label{prp:ode_solutions}
  The solution $x(t)$ of (\ref{system1}) which corresponds to the initial condition
  $x^{(0)}=(x_1^{(0)},\dots,x_n^{(0)})$ is given by
  \begin{equation} \label{solkahan}
    x_i(t)=x_i^{(0)}\frac{(1-f(t)
      h_0)(1+f(t)h_0)}{\(1-f(t)h_0+2f(t)v^{(0)}_{i-1}\)\(1-f(t)h_0+2f(t)v_{i}^{(0)}\)}\;, \quad (i=1,\dots,n)\;,
  \end{equation}
  where $f(t)=\frac{e^{h_0t}-1}{(e^{h_0t}+1)h_0}=\frac{1}{h_0}\tanh(\frac{h_0t}{2})$ when $h_0$ (the value of $H$
  at $x^{(0)}$) is different from zero and $f(t)=t/2$ otherwise.  Also,
  $v_i^{(0)}=a_1x_1^{(0)}+\cdots+a_ix_i^{(0)}.$
\end{proposition}
Notice that when $h_0\neq0$, (\ref{solkahan}) can be rewritten as
\begin{equation*}
  x_i(t)=\frac{x_i^{(0)} e^{th_0}h_0^2}{\(h_0+(e^{th_0}-1) v^{(0)}_{i-1}\)\(h_0+(e^{th_0}-1) v^{(0)}_{i}\)}\;,
  \quad (i=1,\dots,n)\;.
\end{equation*}
\begin{remark}
When several of the parameters $a_i$ in the Hamiltonian function $H$ are equal to zero, so that $H$ is independent
of the corresponding variables $x_i$, the vector field (\ref{eq:GLV_ode}) is a Hamiltonian vector field with
respect to a family of compatible Poisson structures, always with the same Hamiltonian $H$. Indeed, suppose that
$a_i=a_j=0$, with $i< j$. Then, in the computation of the vector field $\dot x_k=\pb{x_k,H}$, $k=1,\dots,n$, the
Poisson brackets $\pb{x_i,x_j}=-\pb{x_j,x_i}$ are not used, so we may replace $\pb{x_i,x_j}=-\pb{x_j,x_i}$ by an
arbitrary function $f_{ij}$ of $x_1,\dots,x_n$ without any effect on the vector field. However, in order for the
new bracket to be a Poisson bracket, it has to satisfy the Jacobi identity, which puts several restrictions on the
function $f_{ij}$. One way to satisfy this restriction is to take $f_{ij}:=a_{ij}x_ix_j$, where $a_{ij}$ is an
arbitrary constant. In fact, replacing $\pb{x_i,x_j}=x_ix_j$ by $\pb{x_i,x_j}=a_{ij}x_ix_j$ for all $i<j$ for which
$a_i=a_j=0$, the new brackets will still be of the general form $\pb{x_i,x_j}=b_{ij}x_ix_j$, known in the
literature as \emph{diagonal brackets}; such brackets are known to automatically satisfy the Jacobi identity
\cite[Example 8.14]{PLV} so they are Poisson brackets. Clearly, any linear combination of these diagonal Poisson
brackets is again a diagonal Poisson bracket, hence all these brackets are compatible. The upshot is that when
$k\geqs 2$ parameters are equal to zero, then (\ref{eq:GLV_ode}) has a multi-Hamiltonian structure: it is
Hamiltonian with respect to a $k\choose 2$-dimensional family of Poisson brackets.
\end{remark}

\section{The Kahan discretization}\label{sec:discrete}
In this section we consider the Kahan discretization of the system (\ref{system1}). Let us recall quickly the
construction of the Kahan discretization of a quadratic vector field $\dot x_i=Q_i(x)$ (see
e.g.\ \cite{MR3002854}). Let $\Phi_i(y,z)$ denote the symmetric bilinear form which is associated to the quadratic
form $Q_i$ and let $\epsilon$ denote a positive parameter, which should be thought of as being small. Then the
\emph{Kahan discretization with step size} $\epsilon$ is the map\footnote{When the map which is defined by the
  discretization is iterated, one often writes it as $x^{(m)}_i\mapsto x_i^{(m+1)}$.} $x_i\mapsto\tilde x_i$,
implicitly defined by
\begin{equation}\label{eq:Kahan_general}
  \tilde{x}_i-x_i=\epsilon\Phi_i(x,\tilde x)\;.
\end{equation}%
We refer to this map as the \emph{Kahan map} (associated to $\dot x_i=Q_i(x)$).  It is well known that the Kahan
map preserves the linear integrals of the initial continuous system (quadratic vector field). So, in our case of
the generalized Lotka-Volterra system, its Hamiltonian function $H=a_1x_1+a_2x_2+\cdots+a_n x_n$ is an invariant of
the Kahan map. As we are going to show in this section the Kahan map (of this system) preserves the Poisson
structure as well; we will also see in the next section that all constants of motion, in particular the ones that
appear in Theorems \ref{thm:integrable}, \ref{thm:integrable_odd} and~\ref{thm:ode_super}, are also invariants of
the Kahan map.

We begin with a lemma which provides an explicit formula for the Kahan discretization of the generalized
Lotka-Volterra system.

\begin{proposition}\label{prop:kahan}
  The Kahan discretization with step size $2\epsilon$ of the system (\ref{system1}) is the rational map
  $\mathcal{K}:(x_1,\dots,x_n) \mapsto (\tilde{x}_1,\dots,\tilde{x}_n)$, given by
  \begin{equation} \label{kahan}
    \tilde{x}_i=x_i\frac{(1-\epsilon H)(1+\epsilon H)}
    {(1-\epsilon H+2\epsilon v_{i-1})(1-\epsilon H+2\epsilon v_{i})}\;, \quad (i=1,\dots,n)\;.
  \end{equation}
\end{proposition}

\begin{proof}
Let us write $\tilde{v}_j=a_1\tilde{x}_1+\cdots+a_j\tilde{x}_j$, in analogy with the functions $v_j$. According to
(\ref{eq:Kahan_general}), the Kahan discretization of (\ref{system2}) (which is equivalent to (\ref{system1})) is
given by
\begin{equation}\label{eq:kahan_v}
  \tilde x_i-x_i=\epsilon x_i(H-\tilde v_i-\tilde v_{i-1})+\epsilon\tilde x_i(H-v_i-v_{i-1})\;,\qquad (i=1,\dots,n)\;,
\end{equation}%
where we have used that $H$ is invariant ($\tilde H=H$).  Summing up these equations, multiplied by $a_i$, for
$i=1,\dots,j$, we get
\begin{equation}\label{eq:kahan_v2}
  \tilde v_j-v_j=\epsilon(v_jH+\tilde v_jH-\delta_j)\;,
\end{equation}%
where $\delta_j$ is given by
$$
  \delta_j:=\sum_{i=1}^ja_ix_i(\tilde v_i+\tilde v_{i-1})+\sum_{i=1}^ja_i\tilde x_i(v_i+v_{i-1})=2v_j\tilde v_j\;.  
$$
The last equality can be proven by an easy recursion on $j$: on the one hand, $\delta_1=a_1x_1\tilde v_1+a_1\tilde
x_1v_1=2v_1\tilde v_1$, while on the other hand
\begin{eqnarray*}
  \delta_{j+1}-\delta_j&=&a_{j+1}x_{j+1}(\tilde v_{j+1}+\tilde v_{j})+a_{j+1}\tilde x_{j+1}(v_{j+1}+v_{j})\\
  &=&2a_{j+1}x_{j+1}\tilde v_j+2a_{j+1}\tilde x_{j+1}v_j+2a_{j+1}^2x_{j+1}\tilde x_{j+1}\;,
\end{eqnarray*}%
and so
\begin{eqnarray*}
 \delta_{j+1}&=&2a_{j+1}x_{j+1}\tilde v_j+2a_{j+1}\tilde x_{j+1}v_j+2a_{j+1}^2x_{j+1}\tilde x_{j+1}+2v_j\tilde v_j\\
              &=&2(v_j+a_{j+1}x_{j+1})(\tilde v_j+a_{j+1}\tilde x_{j+1})=2v_{j+1}\tilde v_{j+1}\;.
\end{eqnarray*}
Solving (\ref{eq:kahan_v2}) (with $\delta_j=2v_j\tilde v_j$) linearly for $\tilde v_j$ we get
\begin{equation} \label{khv}
  \tilde{v}_j=v_j \frac{1+\epsilon H}{1-\epsilon H+2 \epsilon v_j}\;.
\end{equation}
Substituting this into  (\ref{eq:kahan_v}) leads to
$$
  \tilde{x}_i-x_i=
  \epsilon x_i\left(H-v_i\frac{1+\epsilon H}{1-\epsilon H+2\epsilon v_i}-v_{i-1}\frac{1+\epsilon H}{1-\epsilon
    H+2\epsilon v_{i-1}}\right)+\epsilon \tilde{x}_i(H-v_i-v_{i-1})\;,
$$
which can be solved linearly for $\tilde x_i$. It yields the formula (\ref{kahan}).
\end{proof}

\begin{proposition}\label{prop:kahan_poisson}
  The Kahan map $\mathcal{K}$, given by (\ref{kahan}), is a Poisson map with respect to the Poisson bracket $\PB$.
\end{proposition}
\begin{proof}
Recall that the Poisson bracket $\PB$ is given by $\{x_i,x_j\}= x_i x_j$, for $1\leq i < j \leq n$. Therefore, we
need to show that $\{\tilde x_i,\tilde x_j\}= \tilde x_i \tilde x_j$, for $1\leq i < j \leq n$. We set, for
$k=1,\dots,n$,
$$
  A_k=x_k(1-\epsilon H)(1+\epsilon H), \ B_k \ = \ (1-\epsilon H+2\epsilon v_{k-1})(1-\epsilon H+2\epsilon v_k)\;,
$$
so that $\tilde x_k=A_k/B_k$. Then
\begin{equation*}
  \{\tilde{x}_i,\tilde{x}_j\}=\frac{A_iA_j\pb{B_i,B_j}-A_iB_j\pb{B_i,A_j}-B_iA_j\pb{A_i,B_j}+B_iB_j\pb{A_i,A_j}}
    {B_i^2B_j^2}\;.
\end{equation*}
%
%
The Poisson brackets in the right-hand side of this equation can be computed using besides (\ref{pb_vi}) the
following formulas:
\begin{equation*} \label{brxv}
  \{x_i,H\}= x_i(H-v_i-v_{i-1})\;, \qquad \{x_i,v_j \}= \left\{
  \begin{array}{ll}
    \ x_i(v_j-v_i-v_{i-1}) \quad& \mbox{for }    i \leq j\;, \\
    - x_i v_j &\mbox{for } i>j\;.
  \end{array}
  \right.
\end{equation*}
After some computation, it leads to 
\begin{eqnarray*}
  \{\tilde{x}_i,\tilde{x}_j\}&=& \frac{(1-\epsilon^2 H^2)^2 x_i x_j}
          {(1-\epsilon H+2\epsilon v_{i-1})(1-\epsilon H+2\epsilon v_{i})
            (1-\epsilon H+2\epsilon v_{j-1})(1-\epsilon H+2\epsilon v_{j})}\\
  &=&\tilde{x}_{i}\tilde{x}_j\;,
\end{eqnarray*}
as was to be shown.

\end{proof}

An easy comparison of the solution (\ref{solkahan}) to the continuous system and the Kahan map (\ref{kahan}) shows
that the Kahan map is a time advance map for the continuous system, hence preserves all integral curves of the
continuous system and so all constants of motion of the continuous system are invariants for the Kahan
map. Precisely, let $x^{(0)}=(x_1^{(0)},\dots,x_n^{(0)})$ be any point of $\bbR^n$ and let $\epsilon\in\bbR$ be
small but positive. As above, the value of $H$ at $x^{(0)}$ is denoted by $h_0$. Let $t_\epsilon$ denote the unique
solution to the equation $f(t_\epsilon)=\epsilon$, where $f(t)$ is the function given in Proposition
\ref{prp:ode_solutions}. With these notations, (\ref{solkahan}) and (\ref{kahan}) imply that
$x_i(t_\epsilon)=\tilde x_i^{(0)}$. It leads, in view of Theorems \ref{thm:integrable} and
\ref{thm:integrable_odd}, to the following corollary:

\begin{corollary}\label{cor:kahan_integrable}
  The Kahan discretization (\ref{kahan}) is Liouville integrable, with invariants given in Theorem
  \ref{thm:integrable} (resp.\ Theorem \ref{thm:integrable_odd}) when $n$ is even (resp.\ when $n$ is
  odd). It is also superintegrable, with invariants given in Theorem  \ref{thm:ode_super}.
\end{corollary}

Let us denote the $k$-th iterate of the Kahan map (\ref{kahan})  starting from the initial condition
$x^{(0)}=(x_1^{(0)}, \dots, x_n^{(0)})$ by $x^{(k)}$. Then the relation between the solutions to the continuous
system and the Kahan map can be written as $x_i(t_\epsilon)=x_i^{(1)}.$ Now notice that $t_\epsilon$ depends only
on $x^{(0)}$ through $h_0$; this implies that the restriction of $\mathcal{K}$ to the integral curve through
$x^{(0)}$ is the time $t_\epsilon$ flow of the continuous system (restricted to the integral curve through
$x^{(0)}$). Thus, $x^{(2)}$ is obtained from $x^{(1)}$ by the time $t_\epsilon$ flow, and hence from $x^{(0)}$ by
the time $2t_\epsilon$ flow, $x^{(2)}=x(2t_\epsilon)$; more generally, $x^{(m)}$ is obtained from $x^{(0)}$ by the
time $mt_\epsilon$ flow, $x^{(m)}=x(mt_\epsilon)$. It leads to the following proposition.
\begin{proposition}\label{prop:discrete_sol}
  The solution of the discrete system 
  \begin{equation} \label{kahanmapping}
    \tilde {x}_i={x}_i \frac{(1-\epsilon H)(1+\epsilon H)}
                            {(1-\epsilon H+2\epsilon v_{i-1})(1-\epsilon H+2\epsilon v_{i})}\;,\quad(i=1,\dots,n)
  \end{equation}
  with $H=\sum a_ix_i$ and initial condition ${x}^{(0)}$ is given by
  \begin{equation} \label{kahansol2}
    {x}_i^{(m)}={x}_i^{(0)} \frac{(\frac{1+ \epsilon h_0}{1-\epsilon h_0})^m h_0^2}
           {\(h_0+v_{i-1}^{(0)}((\frac{1+ \epsilon h_0}{1-\epsilon h_0})^m-1)\)
           \(h_0+v_{i}^{(0)}((\frac{1+ \epsilon h_0}{1-\epsilon h_0})^m-1)\)}\;.
  \end{equation}
  when $h_0$ (the value of $H$ at $x^{(0)}$) is different from zero. When $h_0=0,$
  \begin{equation}\label{kahansol3}
    x_i^{(m)}=x_i^{(0)}\frac1{\(1+2m\epsilon v^{(0)}_{i-1}\)\(1+2m\epsilon v_{i}^{(0)}\)}\;.
  \end{equation}
\end{proposition}

\begin{proof}
In view of Proposition \ref{prp:ode_solutions},
\begin{eqnarray}\label{eq:iterate}
  x_i^{(m)}&=&x_i(mt_\epsilon)\nonumber\\
  &=&x_i^{(0)}\frac{(1-f(mt_\epsilon) h_0)(1+f(mt_\epsilon)h_0)}
  {\(1-f(mt_\epsilon)h_0+2f(mt_\epsilon)v^{(0)}_{i-1}\)
  \(1-f(mt_\epsilon)h_0+2f(mt_\epsilon)v_{i}^{(0)}\)}\;.\nonumber\\
\end{eqnarray}
When $h_0\neq0$, it follows easily from $f(t)=\frac{e^{th_0}-1}{(e^{th_0}+1)h_0}$ and $f(t_\epsilon)=\epsilon$
that $ e^{t_\epsilon h_0}=\frac{1+h_0\epsilon}{1-h_0\epsilon}.$ In turn, we can compute $f(mt_\epsilon)$ from it,
namely
\begin{equation}\label{eq:fm}
  f(mt_\eps)=\frac1{h_0}\frac{e^{mt_\epsilon h_0}-1}{e^{mt_\eps h_0}+1}=\frac1{h_0}\,
  \frac{\left(\frac{1+h_0\epsilon}{1-h_0\epsilon}\right)^m-1}
      {\left(\frac{1+h_0\epsilon}{1-h_0\epsilon}\right)^m+1}\;.
\end{equation}%
It now suffices to substitute (\ref{eq:fm}) in (\ref{eq:iterate}) and to simplify the resulting expression to
obtain (\ref{kahansol2}). When $h_0=0$, we have that $f(mt_\epsilon)=m\epsilon$, since $f(t)=t/2$. Substituted in
(\ref{eq:iterate}) (with $h_0=0$), we get at once (\ref{kahansol3}).
\end{proof}

\section{Conclusion}\label{sec:comments}

We presented a new class of generalized Lotka-Volterra systems which are, together with their Kahan
discretizations, Liouville integrable and superintegrable, and we provided their explicit solutions. Since linear
Hamiltonians are always preserved under Kahan discretization and since the Poisson structure that we used is
quadratic, it is natural to ask which quadratic Poisson structures on $\bbR^n$ are preserved by the Kahan
discretization of every Hamiltonian vector field with linear Hamiltonian; in view of what we have shown, the
Poisson structure defined by defined by the brackets $\{x_i,x_j\}:= x_i x_j$, for $1\leq i < j \leq n$, belongs to
this class. The Hamiltonian systems which are defined by them would then be good candidates for being Liouville
integrable and/or superintegrable. In view of the recent developments in discretization of polynomial vector fields
by polarization (\cite{CMMOQ}), similar questions can also be considered for higher degree polynomial Hamiltonian
vector fields.

\bigskip\noindent
\emph{Acknowledgement.}\quad {This work was supported by the Australian Research Council. GRWQ is
grateful to Arte\-studioginestrelle and its Director Marina Merli, for
providing the stimulating environment where some of this work was carried
out.}

\end{document}